\documentclass[11pt,british,refpage,intoc,bibliography=totoc,index=totoc,BCOR=7.5mm,captions=tableheading]{extarticle}
\usepackage{mathptmx}
\usepackage[T1]{fontenc}
\usepackage[latin9]{inputenc}
\usepackage[a4paper]{geometry}
\usepackage{color}
\usepackage{babel}
\usepackage{mathrsfs}
\usepackage{mathtools}
\usepackage{amsmath}
\usepackage{amsthm}
\usepackage{amssymb}
\usepackage{enumitem}
\usepackage{graphicx}
\usepackage[colorinlistoftodos]{todonotes}
\usepackage[numbers]{natbib}
\usepackage[unicode=true,pdfusetitle, bookmarks=true,bookmarksnumbered=true,bookmarksopen=false,
 breaklinks=false,pdfborder={0 0 1},backref=false,colorlinks=true]
 {hyperref}
\hypersetup{pdfborderstyle=,linkcolor=black,citecolor=blue,urlcolor=black,filecolor=blue,pdfpagelayout=OneColumn,pdfnewwindow=true,pdfstartview=XYZ,plainpages=false}

\makeatletter
\numberwithin{equation}{section}
\numberwithin{figure}{section}
\theoremstyle{plain}
\newtheorem{thm}{\protect\theoremname}[section]
\theoremstyle{definition}
\newtheorem{defn}[thm]{\protect\definitionname}
\theoremstyle{definition}
\newtheorem*{defn*}{\protect\definitionname}
\theoremstyle{remark}
\newtheorem{rem}[thm]{\protect\remarkname}
\theoremstyle{plain}
\newtheorem{lem}[thm]{\protect\lemmaname}
\theoremstyle{plain}
\newtheorem{cor}[thm]{\protect\corollaryname}

\@ifundefined{date}{}{\date{}}
\usepackage{babel}
\usepackage{caption}

\usepackage{newtxmath}

\usepackage{bbm}
\usepackage[nottoc]{tocbibind}
\allowdisplaybreaks[4]
\usepackage{dsfont}
\DeclareMathAlphabet{\mathcal}{OMS}{cmsy}{m}{n}

\addto\captionsbritish{\renewcommand{\corollaryname}{Corollary}}
\addto\captionsbritish{\renewcommand{\theoremname}{Theorem}}
\addto\captionsenglish{\renewcommand{\corollaryname}{Corollary}}
\addto\captionsenglish{\renewcommand{\theoremname}{Theorem}}
\providecommand{\corollaryname}{Corollary}
\providecommand{\theoremname}{Theorem}

\newtheorem{assumption}{Assumption}

\addto\captionsbritish{\renewcommand{\corollaryname}{Corollary}}
\addto\captionsbritish{\renewcommand{\theoremname}{Theorem}}
\addto\captionsenglish{\renewcommand{\corollaryname}{Corollary}}
\addto\captionsenglish{\renewcommand{\theoremname}{Theorem}}

\makeatother

\providecommand{\corollaryname}{Corollary}
\providecommand{\definitionname}{Definition}
\providecommand{\lemmaname}{Lemma}
\providecommand{\remarkname}{Remark}
\providecommand{\theoremname}{Theorem}

\begin{document}
\global\long\def\Sgm{\boldsymbol{\Sigma}}%
\global\long\def\W{\boldsymbol{W}}%
\global\long\def\H{\boldsymbol{H}}%
\global\long\def\P{\mathbb{P}}%
\global\long\def\Q{\mathbb{Q}}%
\global\long\def\rmd{\mathbb{\,\mathrm{d}}}%
\global\long\def\rd{\mathbb{\mathrm{d}}}%
\global\long\def\R{\mathbb{\mathbb{R}}}%
\global\long\def\E{\mathcal{\mathbb{E}}}%
\global\long\def\ione{\vmathbb1}%
\global\long\def\sF{\mathcal{F}}%
\global\long\def\intr{\int_{\R^{3}}}%
\global\long\def\CC{\mathcal{C}}%

\title{On distributions of velocity random fields in turbulent flows}
\author{Jiawei Li\thanks{Department of Mathematical Sciences, Carnegie Mellon University, Pittsburgh, PA 15213, USA, Email address: $\mathtt{jiaweil4@andrew.cmu.edu}$}, $\;$Zhongmin Qian\thanks{Mathematical Institute, University of Oxford, Oxford OX2 6GG, and Oxford Suzhou Centre for Advanced Research, University of Oxford. Email address: $\mathtt{qianz@maths.ox.ac.uk}$} $\;$and$\;$ Mingrui Zhou\thanks{Mathematical Institute, University of Oxford, Oxford OX2 6GG. Email
address: $\mathtt{mingrui.zhou@maths.ox.ac.uk}$}}
\maketitle

\begin{abstract}
The purpose of the present paper is to derive a partial differential equation (PDE) for the single-time single-point probability density function (PDF) of the velocity field of a turbulent flow. The PDF PDE is a highly non-linear parabolic-transport equation, which depends on two conditional statistical numerics of important physical significance. The PDF PDE is a general form of the classical Reynolds mean flow equation \cite{Reynolds1894}, and is a precise formulation of the PDF transport equation \cite{Pope2000}. The PDF PDE provides us with a new method for modelling turbulence. An explicit example is constructed, though the example is seemingly artificial, but it demonstrates the PDF method based on the new PDF PDE.

\medskip{}


\medskip{}

\end{abstract}

\section{Introduction}

The research on statistical properties of turbulence flows can be traced back to the semi-empirical theories of turbulence in 1920's and 1930's, while the seminal advances in the area include Prandtl \citep{Prandtl1925}, von K\'arm\'an \citep{Karman1930} and Taylor \citep{Taylor1921, Taylor1935}. The goal of statistical fluid mechanics is to provide good descriptions and computational tools for understanding the distributions of the velocity random fields of turbulent fluid flows. Unlike some other unsolved problems in theoretical physics, the equations of motion for fluid dynamics, even for turbulent flows, have been known for over a century. These equations are highly non-linear and non-local partial differential equations, and it is difficult to extract information about the evolution of fluid flows in a deterministic manner. Thus, as a matter of fact, the velocity field of a turbulent flow is better to be considered as a random field arising from either the random initial data or a random external force, or both. To understand the statistics of turbulent flows, it is desired to know, if it is possible, the evolution of some distributional characteristics of fluid flows. The distribution of a random field such as the velocity field is rather complicated and determining the distribution of turbulent flows is a challenging task even when the initial distribution is known. In 1950's  Hopf \citep{EHopf1952} (see also \citep{MoninYaglom1965}) derived a functional differential equation for the law of the velocity random field, but his equation involves functional derivatives.

In the past decades, the probability density function (PDF) method, based on the transport equation, a formal adjoint equation of the Navier-Stokes equation, has been developed into a useful tool for modelling turbulent flows. This method focuses on evaluating the one-point one-time PDF $p(u;x,t)$ of the velocity field $U(x,t)$ or equivalently the centred field $U(x,t)-\overline{U(x,t)}$. The exact transport equation for the PDF, which involves the mean of the pressure term as well as the conditional expectation of the pressure term, has been derived by Pope and can be found in \citep{Pope1985,Pope2000} for details. However, only few features can be extracted from the formal PDF transport equation for the purpose of modelling turbulent flows. Therefore, applications of PDF methods have been based on the generalised Langevin model, where the time-dependent velocity $U(t)$ of a particle at position $X(t)$ is assumed to satisfy a stochastic differential equation (SDE). 

The main contribution in this paper is the derivation of the PDF partial differential equation (PDE) to the velocity random field which is much more explicit than the formal PDF transport equation. This PDE is a generalisation of Reynolds mean flow equations, which can be closed by introducing Reynolds stress tensor field. Having $6$ dimensions in space $(u,x)$ and $1$ dimension in time $t$, our PDF PDE can be regard as a parabolic-transport equation which has a parabolic operator $\frac{1}{2}\partial_{t}+u\cdot\nabla_{x}+\nu\Delta_{x}$ in $x$ and transport operator $\frac{1}{2}\partial_{t}+\nabla_{u}\cdot$ in $u$. However, the PDF PDE for velocity fields is a second order partial differential equation which is in general not parabolic due to the appearance of a mixed derivative term  $\partial_{u^{i}}\partial_{x^{k}}p$. Even this mixed derivative does not appear in the PDF PDE (which is the case for some turbulent flows which will be explained below), the parabolic part in the PDF PDE involves only the variable $x$, and therefore even for this case the PDF PDE is highly degenerate. This feature of the PDF PDE distinguishes itself from the prevalent parabolic PDEs or other types of PDE theories in literature. 

The PDF PDE that we have obtained, relies on two conditional structure functions, which are the conditional average increment 
\[
\rho^{i}(x,y,u,t)=\E\left[U^{i}(y,t)-U^{i}(x,t)\,|\,U(x,t)=u\right]
\]
and the conditional covariance 
\[
\sigma^{ij}(x,y,u,t)=\E\left[U^{i}(y,t)U^{j}(y,t)\;|\;U(x,t)=u\right].
\]
These conditional structure functions describe the interactions of the velocity random field at different positions, hence they are natural to appear in the PDF PDE. The fact that the distribution of velocity random fields is characterised by the conditional first and second moments is an interesting feature reveled  in this paper. These statistical characteristics  are local, which have the capacity of determining the PDF PDE. Moreover, these local statistical characteristics can localise many concepts, such as homogeneity, isotropy and etc, which were introduced firstly by Taylor and Kolmogorov \citep{K41a,K41b,Taylor1935}, allowing us to generalise such concepts to their weak versions.


We outline the main structure of this paper in the following. In section 2, we introduce definitions related to random fields, which are cornerstones of our main results. The evolution equation for the distribution of the velocity random field of turbulence over time will be derived, under the assumption that the random field is regular. The PDE is going to be applied to various types of flows in section 3, including both the viscid and inviscid cases. We also obtain a stochastic representation formula for the solution of the PDF PDE, together with the constraint that ensures the solution is indeed a PDF for all time $t\geq0$ and position $x\in \R^3$. These theoretical results are important when we apply PDF PDE for modelling turbulent flows. Section 4 is thus devoted to an example of modelling the PDF of turbulence, which has the ability of demonstrating the change of distribution at a fixed position $x$ over time. Our paper will be closed by a few remarks in the last section.

\textit{Conventions on notations}. The following set of conventions is employed throughout the paper. Firstly Einstein's convention on summation on
repeated indices through their ranges is assumed, unless otherwise specified. If $A$ is a vector or a vector field (usually in the space of dimension three) dependent on some parameters, then its components are labelled with upper-script indices so that $A \coloneqq \left(A^{i}\right) =\left(A^{1},A^{2},A^{3}\right)$. The same convention applies to coordinates too. The derivative operators $\nabla$ and $\Delta$ are labelled with subscripts to indicate the variable to which the operator is applied, such as $\Delta_{x}\coloneqq\partial_{x^{i}}\partial_{x^{i}}$ and $\nabla_{x}\cdot A\coloneqq \partial_{x^{i}}A^{i}$. Finally, the velocity vector field will be denoted by $U=\left(U^{i}\right)$,
unless we specified.

\section{PDF equation of velocity fields}

In this section, we aim to introduce some fundamental concepts on random fields and to derive the evolution equation for the random velocity field $\{U(x,t)\}_{x,t}$, where $x\in\R^{3},$ $t\geq0$ and $U(x,t)$ takes values in $\R^{3}$.

\subsection{Random fields and their statistical characteristics}

Given a random field $\left\{ U(x,t)\right\} _{x,t}$ on a probability space $(\Omega,\sF,\P)$, $U(x,t)$ is, by definition, an $\mathbb{R}^{3}$-valued random variable for every $x\in\mathbb{R}^{3}$ and $t\geq0$. The law or the distribution of $U(x,t)$ for fixed $x$ and $t$ is a probability measure on the Borel $\sigma$-algebra of $\mathbb{R}^{3}$. The distribution of the random field $U$ consists of, by definition, all possible finite-dimensional marginal joint distributions of 
\[
U(x_{1},t_{1}),\ldots,U(x_{n},t_{n})
\]
where $x_{i}\in\mathbb{R}^{3}$, $t_{i}\geq0$ and any positive integer $n$. For example, by saying that the random field $\{U(x,t):x\in\mathbb{R}^{3},t\geq0\}$ is Gaussian, we refer to the fact that any finite-dimensional marginal joint distribution is a Gaussian distribution, which in particular implies that the marginal distribution of $U(x,t)$ for any $(x,t)$ has a normal distribution. We remark that the converse argument is not true in general.

The most important statistical numerics for understanding a random field is the correlation function of two random variables, which plays the dominant role in the study of turbulence \citep{Batchelor1953,MoninYaglom1965}. In this paper, we however emphasize the use of a few statistical characteristics based on the conditional distribution. Let us introduce these statistical numerics, which we believe are of the most importance.
\begin{defn}
Given a time-dependent random field $\left\{ U(x,t)\right\} _{x,t}$ on a probability space $(\Omega,\sF,\P)$, for $x,y,u\in\mathbb{R}^{3}$ and $t\geq0$, and $i\in \{1,2,3\}$,
\begin{enumerate}[label=\arabic*)]
    \item the \textit{conditional average increment function} $\rho^{i}$ is defined as 
    \begin{equation}\label{eq: con_mean_diff}
        \rho^{i}(x,y,u,t)=\E\left[U^{i}(y,t)-U^{i}(x,t)\,|\,U(x,t)=u\right],
    \end{equation}
    \item the \textit{conditional covariance function} $\sigma(x,y,u,t)$ is defined to be the covariance of $U(y,t)-U(x,t)$ given $U(x,t)=u$,
    \begin{equation}\label{eq: con_covar}
        \sigma^{ij}(x,y,u,t)=\textrm{cov}\left[\left(U^{i}(y,t),U^{j}(y,t)\right)\,|\,U(x,t)=u\right].
    \end{equation}
\end{enumerate}
\end{defn}

From the definition, it is clear that for every $i$, the \textit{conditional mean function} is of the form
\begin{equation}
b^{i}(x,y,u,t)\coloneqq\E\left[U^{i}(y,t)\,|\,U(x,t)=u\right]=\rho^{i}(x,y,u,t)+u^{i}\label{eq: con_mean}
\end{equation}
and $\rho(x,x,u,t)=0$ for all $u$, $x$ and $t$. The conditional covariance function $\sigma^{ij}(x,y,u,t)$ can be treated as the conditional Reynold stress. These statistical characteristics have explicit representations in terms of the two-point joint distribution. For our purpose, it is convenient to assume that the distribution of $U(x,t)$ for every $(x,t)$ has a probability density function (PDF) with respect to Lebesgue measure on $\R^{3}$, denoted by $p(u;x,t)$, in the sense that 
\[
\E\left[\ione_{\left\{ U(x,t)\in E\right\} }\right]=\int_{E}p(u;x,t)\rmd u\quad\textrm{ for any Borel set }E\in\mathscr{B}({\R^{3}}).
\]

Similarly, the joint distribution of $U(x,t)$ and $U(y,t)$ at two distinct points $x\neq y$ has a joint PDF, denoted by $p_{2}(u_{1},u_{2};x,y,t)$. It follows that the conditional law of $U(y,t)$ given that $U(x,t)=u$ possesses the following conditional PDF 
\[
p_{2|1}(v;u,x,y,t)\coloneqq\P(U(y,t)=v\,|\:U(x,t)=u)=\frac{p_{2}(u,v;x,y,t)}{p(u;x,t)}
\]
with $p_{2|1}(v;u,x,y,t)=0$ if $p(u;x,t)=0$. In terms of the conditional law, the joint PDF of $U(x,t)$ and $U(y,t)$ may be split into a product 
\begin{align*}
p_{2}(u_{1},u_{2};x,y,t) & =p(u_{1};x,t)p_{2|1}(u_{2};u_{1},x,y,t)\\
 & =p(u_{2};y,t)p_{2|1}(u_{1};u_{2},y,x,t).
\end{align*}
As a result, we are allowed to represent the conditional average difference function (\ref{eq: con_mean_diff}) and covariance function (\ref{eq: con_covar}) as an integral relevant to the conditional density, namely
\begin{equation}
\rho^{i}(x,y,u,t)=\intr(v^{i}-u^{i})p_{2|1}(v;u,x,y,t)\rmd v\label{eq: con_mean_diff_integral}
\end{equation}
and 
\[
\sigma^{ij}(x,y,u,t)=\intr\left(v^{i}-b^{i}\right)\left(v^{j}-b^{j}\right)p_{2|1}(v;u,x,y,t)\rmd v.
\]

The use of the conditioning techniques is in fact the main reason for advocating the foundation of statistical fluid mechanics based on the probability theory, rather than on an average procedure, which was first explicitly proposed by Kolmogorov \citep{K41a,K41b}. The homogeneity and the isotropy can be defined in general for random fields indexed by a space variable $x\in\mathbb{R}^{3}$, which have been introduced into the study of turbulence by G. I. Taylor. The local homogeneous and local isotropic flows were introduced by Kolmogorov for formulating K41 theory (and its improved version K61 theory). According to Kolmogorov \citep{K41a,K41b}, a random field $\left\{ U(x,t)\right\} _{x,t}$ is locally homogeneous if for any $x,y\in \R^3$, the conditional distribution of $U(y,t)-U(x,s)$ given $U(x,s)=u$ depends on $y-x$ and $u$, and further it is locally isotropic if the conditional distribution depends only on $|y-x|$ and $u$. By using the conditional average and the conditional covariance functions, it is possible to generalise these terminologies to their weak versions. We are now in a position to state technical assumptions on the random field.
\begin{defn}
The random field $\{U(x,t)\}_{x,t}$ on the probability space $(\Omega,\sF,\P)$ is 
\begin{enumerate}[label=\arabic*)]
    \item \textit{regular} if the conditional average increment function $\rho$ has derivatives up to second order and $\rho(x,y,u,t)$ has a Taylor expansion (for every $x,u$ and $t$ fixed) about $y$: 
    \begin{equation}
        \rho(x,y,u,t)=B_{k}(x,u,t)(y^{k}-x^{k})+\frac{1}{2}A_{jk}(x,u,t)(y^{j}-x^{j})(y^{k}-x^{k})+o\left(|y-x|^{2}\right)\label{eq: taylor_expansion}
    \end{equation}
    as $|y-x|\rightarrow0$, where 
    \begin{equation}
        B_{k}(x,u,t)=\left.\frac{\partial}{\partial y^{k}}\rho(x,y,u,t)\right|_{y=x}\textrm{ and }A_{jk}(x,u,t)=\left.\frac{\partial^{2}}{\partial y^{j}\partial y^{k}}\rho(x,y,u,t)\right|_{y=x},\label{eq: BA_vectors}
    \end{equation}
    and $B$ and $A$ are differentiable in $x^{i},u^{i},t$ for all $i\in\{1,2,3\}$;
    \item \textit{weakly homogeneous} if given $u\in\R^{3},t\geq0$, $\rho^{i}(x,y,u,t)=o(|y-x|)$ as $|y-x|\rightarrow0$ for all $i,y,x$;
    \item \textit{weakly isotropic} if both $\rho(x,y,u,t)$ and $\sigma(x,y,u,t)$ depend only on $|y-x|$, $u$ and $t$, and $A_{kk}^{i}$ only depends on $u$ and $t$.
\end{enumerate}
\end{defn}

The functions $B_{k}^{i}$ and $A_{jk}^{i}$ in the Taylor expansion of the conditional average increment $\rho^{i}(x,y,u,t)$ also have the form
\[
\begin{cases}
B_{k}^{i}(x,u,t)\,=\E\left[\frac{\partial}{\partial x^{k}}U^{i}(x,t)\,\bigg|\,U(x,t)=u\right],\\
A_{jk}^{i}(x,u,t)=\E\left[\frac{\partial^{2}}{\partial x^{j}\partial x^{k}}U^{i}(x,t)\,\bigg|\,U(x,t)=u\right]
\end{cases}
\]
for $i,j,k\in\{1,2,3\}$. Moreover, if the $\{U(x,t)\}_{x,t}$ is weakly homogeneous, we have the equivalent characterisation
\[
B_{k}^{i}(x,u,t)=\lim_{\varepsilon\rightarrow0}\frac{\rho^{i}(x,x+\varepsilon e^{(k)},u,t)}{\varepsilon}=0
\]
for all $i,k\in\{1,2,3\}$, where $e^{(k)}$ is the unit vector at the
$k$-th direction. 

Unlike Kolmogorov's definition of isotropic flows and homogeneous flows, our concept on weakly isotropy has no direct relationship to weakly homogeneity. Nevertheless, a regular locally homogeneous turbulent flow in the sense of Kolmogorov satisfies the condition that $A_{kk}^{i}$ depends only on $u$, since the conditional average increment of such a flow must obey $\rho^{i}(x,y,u,t)=g^{i}(y-x,u)$ for some function $g^{i}$ and
\begin{align*}
A_{jk}^{i}(x,u,t) & =\frac{\partial^2 g^{i}}{\partial y^{j}\partial y^{k}}(0,u).
\end{align*}
Moreover, if we further assume the flow is locally isotropic, the conditional average increment function satisfies $\rho^{i}(x,y,u,t)=g^{i}(|y-x|,u)$ and 
\begin{align*}
B_{k}^{i}(x,u,t)=\lim_{\varepsilon\rightarrow0^{+}}\frac{\rho^{i}(x,x+\varepsilon e^{(k)},u,t)}{\varepsilon} & =\lim_{\varepsilon\rightarrow0^{+}}\frac{-\rho^{i}(x,x-\varepsilon e^{(k)},u,t)}{\varepsilon}
\end{align*}
is well-defined if and only if $B_{k}^{i}=0$. Therefore, this turbulent flow is both weakly homogeneous and weakly isotropic in our sense.

Apart from extending Kolmogorov's definitions on homogeneity and isotropy, the significance of introducing these terminologies is they will eliminate the mixed-terms in the PDF PDE, which will be thoroughly explained in section \ref{sec: Application_to_turbulent_flows}.

\subsection{The evolution equation for the velocity distribution}

In this subsection, we derive the main theoretical result, which provides the theoretical foundation of modelling PDFs of turbulent flows based on two statistical characteristics. We consider an incompressible turbulent flow, inviscid or viscous, with kinetic viscosity constant $\nu$ which is positive for viscous fluid, or $\nu$ reads as zero for inviscid fluid. The turbulent flow is described by its velocity $U=(U^{1},U^{2},U^{3})$ and the pressure $P$, which are random fields and the three dimensional Navier-Stokes equations 
\begin{align}
\frac{\partial U^{i}}{\partial t}+U^{j}\frac{\partial U^{i}}{\partial x^{j}} & =\nu\Delta_{x}U^{i}-\frac{\partial P}{\partial x^{i}},\label{eq:ns-m1}\\
U(x,0) & =U_{0}(x),\nonumber 
\end{align}
where $i=1,2,3$ and $\nu\geq0$ is the viscosity constant, together
with the constraint
\begin{equation}
\frac{\partial U^{j}}{\partial x^{j}}=0.\label{eq:ns-m2}
\end{equation}
The initial condition is also treated as a random field on $\R^{3}$ and each sample path $\omega\in\Omega$ corresponds to a deterministic function $U(x,t;\omega)$, which serves as a solution to equation (\ref{eq:ns-m1}) with initial data $U_{0}(x;\omega)$. We will discuss an ideal case, for the purpose of understanding the local properties of turbulent flows, where the region occupied by the fluid is the entire space $\mathbb{R}^{3}$. Moreover, without further qualifications, the dynamical variables such as $U(x,t)$ and $P(x,t)$ decay to zero sufficiently fast as $|x|$ tends to infinity. In addition, to avoid technical difficulties, but not in any way implying that these issues are not important, we will assume that the dynamical variables $U(x,t)$ and $P(x,t)$ are sufficiently smooth functions of $(x,t)$. 

Due to the divergence-free condition (\ref{eq:ns-m2}) , the pressure term satisfies the following Poisson equation
\[
\Delta_{x}P=-\frac{\partial^{2}(U^{i}U^{j})}{\partial x^{j}\partial x^{i}}.
\]
Therefore, according to the Green formula, we have the integral representation
\begin{equation}
P(x,t)=\intr\frac{1}{4\pi|y-x|}\frac{\partial^{2}(U^{i}U^{j})}{\partial y^{j}\partial y^{i}}\rmd y,\label{eq: pressue}
\end{equation}
which implies in particular that the distribution of $P$ is completely determined by the distribution of the velocity random field.

We assume that $\{U(x,t)\}_{x,t}$ is a regular random field. Since $U(x,t)$ is divergence-free as in equation (\ref{eq:ns-m2}), we have for all $i,k$,
\[
\frac{\partial}{\partial y^{i}}\rho^{i}(x,y,u,t)=0,\quad B_{i}^{i}=0,\quad A_{ik}^{i}=A_{ki}^{i}=0,
\]
as well as the following integral condition for PDF of $U(x,t)$ 
\begin{equation}
\frac{\partial}{\partial x^{i}}\intr u^{i}p(u;x,t)\rmd u=0,\label{eq:div-PDF}
\end{equation}
which will appear as a natural constraint for the PDF PDE we will derive. 

Recall that the Reynolds equation (see \cite{Reynolds1894}) is obtained by taking the average in (\ref{eq:ns-m1}), more explicitly 
\begin{align*}
\frac{\partial\E\left[U^{i}\right]}{\partial t}+\E\left[U^{j}\frac{\partial U^{i}}{\partial x^{j}}\right] & =\nu\Delta_x\E\left[U^{i}\right]-\frac{\partial\E\left[P\right]}{\partial x^{i}},\\
\frac{\partial\E\left[U^{i}\right]}{\partial x^{i}} & =0.
\end{align*}
The conventional treatment for the non-linear term on the left-hand side is to write 
\[
\E\left[U^{j}\frac{\partial U^{i}}{\partial x^{j}}\right]=\E\left[U^{j}\right]\frac{\partial\E\left[U^{i}\right]}{\partial x^{j}}+\frac{\partial}{\partial x^{i}}r^{ij},
\]
where 
\[
r^{ij}=\E\left[(U^{i}-\E[U^{i}])(U^{j}-\E[U^{j}])\right]
\]
is the Reynolds stress. The PDF equation can be obtained by carrying out this computation for the average $\E\left[F(U)\right]$ where $F:\R^{3}\to\R$ is set as a smooth function with a compact support, instead of choosing $F(x)=x$ for each velocity component in the case of Reynold stress. We are now in a position to establish the most important work in this
paper:
\begin{thm}
\label{thm: PDF_PDE}Let $\left\{ U(x,t)\right\} _{(x,t)\in\R^{3}\times\R_{+}}$ be a regular random field defined on the probability space $(\Omega,\sF,\P)$ and take values in $\R^{3}.$ Suppose $\left\{ U(x,t)\right\} _{x,t}$ satisfies the Navier-Stokes equation (\ref{eq:ns-m1}) and the continuity equation (\ref{eq:ns-m2}), then the PDF $p(u;x,t)$ of the velocity $U(x,t)$ satisfies the evolution equation: 
\begin{align}
\begin{split}\label{eq: PDF_PDE_origin}
    \frac{\partial p}{\partial t}+u^{i}\frac{\partial p}{\partial x^{i}} & =\nu\Delta_{x}p+\frac{\partial}{\partial u^{i}}\left(\nu\frac{\partial\left(pB_{k}^{i}\right)}{\partial x^{k}}-\nu pA^{i}+pQ^{i}\right),\\
p(u;x,0) & =p_{0}(u;x),
\end{split}
\end{align}
where $p_{0}(u;x)$ is the PDF of the $U_{0}(x)$ in the random field
$\left\{ U_{0}(x)\right\} _{x}$ and 
\begin{equation}
Q^{i}(x,u,t)=\intr\frac{y^{i}-x^{i}}{4\pi|y-x|^{3}}\frac{\partial^{2}\left(\sigma^{jk}+b^{j}b^{k}\right)}{\partial y^{k}\partial y^{j}}\rmd y,\label{eq:q-term}
\end{equation}
\begin{equation}
B_{k}^{i}=\left.\frac{\partial\rho^{i}(x,y,u,t)}{\partial y^{k}}\right|_{y=x},\textrm{ and }A^{i}(x,u,t)=\left.\Delta_{y}\rho^{i}(x,y,u,t)\right|_{y=x}=A_{kk}^{i}(x,u,t).\label{eq: BA_in_PDE}
\end{equation}
\end{thm}

\begin{proof}
Let $F\in D(\R^{3})$ be a test function, which is a smooth function taking values in $\R$ with a compact support. For simplicity, we denote $F_{i}(u)\coloneqq\partial_{u^{i}}F(u)\in D(\R^{3})$. Applying $\partial_{t}$ on the average $\E\left[F(u)\right]$ followed by exchanging integral and their derivative operator, we have
\[
\frac{d}{dt}\mathbb{E}\left[F(U(x,t))\right]=\int F(u)\frac{\partial}{\partial t}p(u;x,t)\rmd u.
\]
On the other hand 
\begin{equation}
\frac{d}{dt}\mathbb{E}\left[F(U(x,t))\right]=\mathbb{E}\left[\frac{\partial}{\partial t}F(U(x,t))\right]\label{eq:exp-tu1}
\end{equation}
together with utilizing the Navier-Stokes equations (\ref{eq:ns-m1}), we get
\[
\left(\frac{\partial}{\partial t}-\nu\Delta_{x}\right)F(U)=-\frac{\partial\left(U^{i}F(U)\right)}{\partial x^{i}}-\nu\frac{\partial F_{i}(U)}{\partial x^{k}}\frac{\partial U^{i}}{\partial x^{k}}-F_{i}(U)\frac{\partial P}{\partial x^{i}}.
\]
Substituting this into (\ref{eq:exp-tu1}), we obtain that 
\[
\mathbb{E}\left[\frac{\partial}{\partial t}F(U)\right]=\nu\Delta_{x}\mathbb{E}\left[F(U)\right]-\frac{\partial}{\partial x^{i}}\mathbb{E}\left[U^{i}F(U)\right]+J_{1}+J_{2},
\]
where the first two terms are equivalent to
\[
\nu\Delta_{x}\mathbb{E}\left[F(U)\right]=\intr F(u)\nu\Delta_{x}p(u;x,t)\rmd u,
\]
and
\[
-\frac{\partial}{\partial x^{i}}\mathbb{E}\left[U^{i}F(U)\right]=\intr F(u)\left(-u^{i}\frac{\partial}{\partial x^{i}}p(u;x,t)\right)\rmd u.
\]
Subsequently, the remaining $J_{1}$, $J_{2}$ are of the form
\[
J_{1}=-\nu\mathbb{E}\left[\frac{\partial F_{i}(U)}{\partial x^{k}}\frac{\partial U^{i}}{\partial x^{k}}\right],\quad J_{2}=-\mathbb{E}\left[F_{i}(U)\frac{\partial P}{\partial x^{i}}\right].
\]

The evaluation of $J_{1}$ and $J_{2}$ requires invoking the joint distribution at two points together with taking limits. Here we depart from this approach by expressing this term via the PDF, which allows us to perform similar computations for a general case. The partial derivative $\frac{\partial U^{i}}{\partial x^{j}}$ may be written as a limit
\[
U^{j}(x,t)\frac{\partial}{\partial x^{j}}U^{i}(x,t)=\lim_{h\rightarrow0}\frac{1}{h}U^{j}(x,t)\left(U^{i}(x+he^{(j)},t)-U^{i}(x,t)\right).
\]
Assuming that we are able to take the average under the limit i.e. the dominated convergence theorem can be applied, we are able to write the non-linear term in terms of 
\[
\E\left[U^{j}\frac{\partial U^{i}}{\partial x^{j}}\right]=\lim_{h\rightarrow0}\frac{1}{h}\E\left[U^{j}(x,t)\left(U^{i}(x+he^{(j)},t)-U^{i}(x,t)\right)\right]=:\lim_{h\rightarrow0}L^{i}(h).
\]
The average appearing on the right-hand side, denoted by $L^{i}(h)$, may be evaluated in terms of the two-point joint distribution 
\begin{align*}
L^{i}(h) & =\intr\intr u^{j}u_{1}^{i}p_{2}(u,u_{1};x,x+he^{(j)},t)\rmd u_{1}\rmd u-\int u^{j}u^{i}p(u;x,t)\rmd u\\
 & =\int u^{j}p(u;x,t)\left[\intr u_{1}^{i}p_{2|1}(u_{1};u,x,x+he^{(j)},t)\rmd u_{1}\right]\rmd u\\
 & \quad -\intr u^{j}u^{i}p(u;x,t)\rmd u\\
 & =\intr u^{j}p(u;x,t)\left[\intr(u_{1}^{i}-u^{i})p_{2|1}(u_{1};u,x,x+he^{(j)},t)\rmd u_{1}\right]\rmd u\\
 & =\intr u^{j}p(u;x,t)\rho^{i}(x,x+he^{(j)},u,t)\rmd u,
\end{align*}
which follows that 
\[
\frac{\partial}{\partial x^{j}}\E\left[U^{j}U^{i}\right]=\intr u^{j}B_{j}^{i}(x,u,t)p(u;x,t)\rmd u.
\]
Now we deal with $J_{1}$. Writing the space derivatives $\frac{\partial F_{i}(U)}{\partial x^{k}}\frac{\partial U^{i}}{\partial x^{k}}$
as the following limits 
\[
\lim_{h\rightarrow0}\frac{1}{h^{2}}\left(F_{i}(U(x+he^{(k)},t))-F_{i}(U(x,t))\right)\left(U^{i}(x+he^{(k)},t)-U^{i}(x,t)\right),
\]
where $e^{(1)}=(1,0,0)$ and so on. Taking expectation first, we obtain
that 
\begin{align*}
J_{1} & =-\nu\lim_{h\rightarrow0}\frac{1}{h^{2}}\intr\left(F_{i}(u_{2})-F_{i}(u_{1})\right)\left(u_{2}^{i}-u_{1}^{i}\right)p_{2}(u_{1},u_{2};x,x+he^{(k)},t)\rmd u_{1}\rmd u_{2}\\
 & =\nu\intr F_{i}(u)\lim_{h\rightarrow0}\frac{I^{i,k}(h)}{h^{2}}\rmd u,
\end{align*}
where 
\[
I^{i,k}(h)\coloneqq\int_{\R^{3}}(u_{1}^{i}-u^{i})p_{2}(u_{1};u,x,x+he^{(k)},t)\mathrm{d}u_{1}+\int_{\R^{3}}\int_{\R^{3}}(u_{1}^{i}-u^{i})p_{2}(u;u_{1},x,x+he^{(k)},t)\mathrm{d}u_{1}.
\]
Using the conditional probability notation that we introduced, we
may write this as 
\begin{align*}
I^{i,k}(h)\coloneqq & p(u;x+he^{(k)},t)\intr(u_{1}^{i}-u^{i})p_{2|1}(u_{1};u,x+he^{(k)},x,t)\rmd u_{1}\\
 & +p(u;x,t)\intr(u_{1}^{i}-u^{i})p_{2|1}(u_{1};u,x,x+he^{(k)},t)\rmd u_{1}\\
= & p(u;x+he^{(k)},t)\left(\rho^{i}(x+he^{(k)},x,u,t)+\rho^{i}(x,x+he^{(k)},u,t)\right)\\
 & -\left(p(u;x+he^{(k)},t)-p(u;x,t)\right)\rho^{i}(x,x+he^{(k)},u,t).
\end{align*}
The last equality is a result of applying (\ref{eq: con_mean_diff_integral}), which converts integrals into conditional average increments $\rho^{i}$. As a consequence of the regularity condition on the random field, we make use of (\ref{eq: taylor_expansion}) to deduce
\begin{align*}
\rho^{i}(x+he^{(k)},x,u,t)&+\rho^{i}(x,x+he^{(k)},u,t)  =\left\{ B_{k}^{i}(x,u,t)h+\frac{1}{2}A_{kk}^{i}(x,u,t)h^{2}+o(h^{2})\right\} \\
&+\left\{ -B_{k}^{i}(x+he^{(k)},u,t)h+\frac{1}{2}A_{kk}^{i}(x+he^{(k)},u,t)h^{2}+o(h^{2})\right\} \\
 & \qquad\qquad\qquad\qquad\;\,=\left(B_{k}^{i}(x,u,t)-B_{k}^{i}(x+he^{(k)},u,t)\right)h\\
 & \qquad\qquad\qquad\qquad\;\,\quad+\frac{1}{2}\left(A_{kk}^{i}(x,u,t)+A_{kk}^{i}(x+he^{(k)},u,t)\right)h^{2}+o(h^{2})
\end{align*}
and therefore
\begin{align*}
\lim_{h\rightarrow0}\frac{1}{h^{2}}I^{i,k}(h) & =-p\frac{\partial}{\partial x^{k}}B_{k}^{i}+pA^{i}-\frac{\partial p}{\partial x^{k}}B_{k}^{i}\\
 & =-\frac{\partial}{\partial x^{k}}(pB_{k}^{i})+pA^{i}.
\end{align*}
We perform integration by parts to derive
\[
J_{1}=\intr F(u)\left[-\frac{\partial}{\partial u^{i}}\left(\nu\left[-\frac{\partial}{\partial x^{k}}(pB_{k}^{i})+pA^{i}\right]\right)\right]\rmd u.
\]
Next we handle $J_{2}$. Applying the representation (\ref{eq: pressue}),
we arrive at 
\[
\frac{\partial P}{\partial x^{i}}=\intr\frac{y^{i}-x^{i}}{4\pi|y-x|^{3}}\frac{\partial^{2}(U^{j}U^{k})}{\partial y^{k}\partial y^{j}}\rmd y,
\]
which implies 
\[
J_{2}=-\mathbb{E}\left[F_{i}(U)\int_{\R^{3}}\frac{y^{i}-x^{i}}{4\pi|y-x|^{3}}\frac{\partial^{2}(U^{j}U^{k})}{\partial y^{k}\partial y^{j}}\rmd y\right].
\]
Writing the derivative in terms of 
\begin{align*}
\frac{\partial^{2}(U^{j}U^{k})}{\partial y^{k}\partial y^{j}} & =\lim_{h\rightarrow0}\frac{1}{h^{2}}\bigg\{ U^{j}(y+h(e^{(k)}+e^{(j)}),t)U^{k}(y+h(e^{(k)}+e^{(j)}),t)\\
 & \qquad\qquad\;-U^{j}(y+he^{(k)},t)U^{k}(y+he^{(k)},t)-U^{j}(y+he^{(j)},t)U^{k}(y+he^{(j)},t)\\
 & \qquad\qquad\;+U^{j}(y,t)U^{k}(y,t)\bigg\}
\end{align*}
and using the two-point PDF by integrating then taking limit as $h\rightarrow0$,
lead us to
\[
J_{2}=-\int_{\R^{3}}F_{i}(u)\left[\int_{\R^{3}}\frac{y^{i}-x^{i}}{4\pi|y-x|^{3}}\frac{\partial^{2}}{\partial y^{k}\partial y^{j}}J_{2}^{jk}\rmd y\right]\rmd u,
\]
where the integral $J_{2}^{jk}$ has the following integral form
\begin{align*}
J_{2}^{jk}: & =\int_{\R^{3}}u_{1}^{j}u_{1}^{k}p_{2}(u,u_{1};x,y,t)\rmd u_{1}\\
 & =p(u;x,t)\intr u_{1}^{j}u_{1}^{k}p_{2|1}(u_{1};u,x,y,t)\rmd u_{1}\\
 & =p(u;x,t)\left(\sigma^{jk}+b^{j}b^{k}\right)(x,y,u,t),
\end{align*}
through using equations (\ref{eq: con_mean_diff}, \ref{eq: con_covar}).
Therefore, substituting this into the equation for $J_{2}$ yields
\begin{align*}
J_{2} & =-\int_{\R^{3}}F_{i}(u)p(u;x,t)\int_{\R^{3}}\frac{y^{i}-x^{i}}{4\pi|y-x|^{3}}\frac{\partial^{2}\left(\sigma^{jk}+b^{j}b^{k}\right)}{\partial y^{k}\partial y^{j}}\rmd y\rmd u\\
 & =\int_{\R^{3}}F(u)\frac{\partial}{\partial u^{i}}\left[p(u;x,t)Q^{i}(x,u,t)\right]\rmd u.
\end{align*}

Putting all terms together, we deduce that
\begin{align*}
\intr F(u) & \left(-\frac{\partial}{\partial u^{i}}\left(\nu\frac{\partial\left(pB_{k}^{i}\right)}{\partial x^{k}}(x,u,t)-\nu p(u;x,t)A^{i}(x,u,t)+p(u;x,t)Q^{i}(x,u,t)\right)\right.\\
 & \quad\left.+\left(\partial_{t}+u\cdot\nabla_{x}-\nu\Delta_{x}\right)p(u;x,t)\right)\rmd u=0
\end{align*}
for all such $F\in D(\R^{3})$. Therefore, we must have (\ref{eq: PDF_PDE_origin}).
\end{proof}
We finish this section by adding several comments. The PDF PDE (\ref{eq: PDF_PDE_origin}) may be written as 
\begin{equation}
\left(\frac{\partial}{\partial t}+u^{i}\frac{\partial}{\partial x^{i}}-\nu\Delta_{x}\right)p=\frac{\partial}{\partial u^{i}}\left(\nu B_{k}^{i}\frac{\partial p}{\partial x^{k}}+pC^{i}\right),\label{eq: PDF_PDE_origin_C}
\end{equation}
where for simplicity we introduce the following vector field 
\begin{equation}
C^{i}=\nu\frac{\partial B_{k}^{i}}{\partial x^{k}}-\nu A^{i}+Q^{i}\label{eq:C-vector}
\end{equation}
for $i=1,2,3$. The equation (\ref{eq: PDF_PDE_origin_C}) is a mixed
type of parabolic and transport PDE. The parabolic operator in variables
$(x,t)$
\[
\frac{\partial}{\partial t}+u^{i}\frac{\partial}{\partial x^{i}}-\nu\Delta_{x}
\]
is independent of fluid flows, which is a significant feature.

Although the PDF PDE (\ref{eq: PDF_PDE_origin_C}) appears to be linear in the PDF $p(u;x,t)$, it is much more complicated than it looks. In particular, the coefficients $A,B$ and $Q$ are functionals of the conditional average and covariance functions, which are in general not determined by the PDF $p(u;x,t)$ alone. Therefore the PDF PDE (\ref{eq: PDF_PDE_origin_C}) is not a closed partial differential equation. The significance of the PDF PDE lies in the fact that if the statistical numerics $\rho$ and $\sigma$ are considered as given, which will be the case for modelling turbulent flows, then the PDF PDE is a partial differential equation of second order, though mixed type of parabolic and transport in general. 

Nevertheless, the PDE (\ref{eq: PDF_PDE_origin}) is a challenging obstacle even if $A,B,Q$ are all considered as given. The function $B$ can be understood as the mean velocity gradient at $(x,t)$ condition on the velocity vector at $(x,t)$, which brings the mixed derivatives $\partial_{u^{i}}\partial_{x^{k}}p$, while the corresponding diffusion matrix $(D_{ij})_{1\leq i,j\leq6}$ collecting the second order terms is of the form
\begin{align*}
D_{ij} & =\begin{cases}
0 & 1\leq i,j\leq3,\\
\nu B_{j-3}^{i} & 1\leq i\leq3,\,4\leq j\leq6,\\
\nu B_{i-3}^{j} & 1\leq j\leq3,\,4\leq i\leq6,\\
\nu\ione_{\left\{ j=k\right\} } & 4\leq i,j\leq6,
\end{cases}
\end{align*}
if we consider $(u,x)$ as a whole. The matrix $D_{ij}$ is not necessarily symmetric and not non-negative definite even if it is symmetric. It poses a challenging mathematical problem developing a theory of this kind of mixed type partial differential equations to facilitate the modelling of turbulent flows based on the PDF PDE. 

\section{Application to turbulent flows\label{sec: Application_to_turbulent_flows}}

As we have seen, our PDE (\ref{eq: PDF_PDE_origin}) does not fit into any existing categories of PDE theories. However, functionals $A,B,Q$ will be derived, if the conditional statistics can be obtained or estimated through practical experiments. Therefore, tracking the PDF of the turbulent flow is equivalent to measuring or modelling the conditional mean and conditional variance, followed by solving the PDF PDE (\ref{eq: PDF_PDE_origin}) using some numerical methods. This brings a new approach on the modelling of turbulent flows. 

In this part, we establish some mathematical tools for the purpose of modelling turbulence based on the PDF PDE.

For convenience, let us introduce the following technical assumptions on a function $f$.
\begin{assumption}\label{assumption: LipschitzCondition}
For a function $f(x,u,t)$ which is uniformly continuous in $t$, there exist constants $K_{1},K_{2}\geq 0$ such that 
\begin{align*}
|f(x,u,t)-f(y,w,t)| & \leq K_{1}(|x-y|+|u-w|),\\
|f(x,u,t)| & \leq K_{2}\left(1+|u|\right),
\end{align*}
for all $x,y,u,w\in\R^{3}$ and $t\geq0$.
\end{assumption}
\begin{assumption}\label{assumption: Polygrowth}
The function $f(u;x)$ is continuous and has up to polynomial growth in $(u,x)$.
\end{assumption}

\subsection{Weakly homogeneous and weakly isotropic flows}

When the viscous incompressible flow is weakly homogeneous, the mix-derivative term disappears and the PDF PDE is simplified to 
\begin{align}
\begin{split}\label{eq: PDF_PDE_homogeneous}
    \left(\frac{\partial}{\partial t}+u^{i}\frac{\partial}{\partial x^{i}}-\nu\Delta_{x}\right)p & =\frac{\partial}{\partial u^{i}}\left(pC^{i}\right),\\
p(u;x,0) & =p_{0}(u;x),
\end{split}
\end{align}
where $C=Q-\nu A$, $A$ and $Q$ are given by equations (\ref{eq:q-term}) and (\ref{eq: BA_in_PDE}) respectively. By the definition of $B$, a weakly homogeneous flow has the property that the velocity gradient condition on the velocity vector on the same location is a centred random vector. The weak homogeneity allows us to state the representation formula, which provides a useful tool when we model weakly homogeneous turbulent flows.
\begin{thm}\label{thm: weakly_homogeneous}
Given the explicit form of $C$, we suppose that $C$ satisfies assumption \normalfont{\textbf{\ref{assumption: LipschitzCondition}}} and $p_{0}(u;x)$ satisfies assumption \normalfont{\textbf{\ref{assumption: Polygrowth}}} respectively. 
\begin{enumerate}[label=\arabic*)]
    \item Assuming $p(u;x,t)$ is the solution to equation (\ref{eq: PDF_PDE_homogeneous}) and also a member of $\CC^{1,2,1}(\R^{3}\times\R^{3}\times[0,T])$ with some fixed $T>0$, the solution $p$ is unique and possesses the following representation form
    \begin{equation}\label{eq: S_rep_h}
        p(u;x,t)=\mathbb{E}\left[p_{0}(Y(t);X(t))q(t)\right],
    \end{equation}
    where for any given $t\in[0,T]$ and $(x,u)$, $(X,Y,q)$ is the unique strong solution to the system of SDEs 
    \[
    \rd X^{i}(s)=-Y^{i}(s)\rmd s+\sqrt{2\nu}\rmd M^{i}(s),\quad X(0)=x,
    \]
    \[
    \rd Y^{i}(s)=C^{i}(X(s),Y(s),t-s)\rmd s,\quad Y(0)=u,
    \]
    and 
    \[
    \rd q(s)=q(s)\frac{\partial C^{k}}{\partial u^{k}}(X(s),Y(s),t-s)\rmd s,\quad q(0)=1,
    \]
    for $i=1,2,3$ and $s\in[0,t]$, associated with $M=(M^{1},M^{2},M^{3})$ being the Brownian motion in $\mathbb{R}^{3}$ with $M(0)=0$ defined on some probability space.
    \item\label{claim: viscosity_solution} The equation (\ref{eq: S_rep_h}) is a unique viscosity solution to the PDF PDE (\ref{eq: PDF_PDE_homogeneous}).
\end{enumerate}
\end{thm}

\begin{proof}
If $C$ is Lipschitz, the previous system of SDEs for $(X,Y)$ has a unique solution $(X,Y)$ and $q$ is given via the exponential function. Notice $(X,Y,q)$ depends on $(x,u)$ as well. Let $\theta(s)=(Y(s);X(s),t-s)$, $\eta(s)=(X(s),Y(s),t-s)$ and denote $Z(s)=p(\theta(s))q(s)$. According to It\^o's formula,
\begin{align*}
\rd Z(s) & =q(s)\frac{\partial p}{\partial u^{i}}(\theta(s))\rd Y^{i}(s)+q(s)\frac{\partial p}{\partial x^{i}}(\theta(s))\rd X^{i}(s)+q(s)\nu\Delta_{x}p(\theta(s))\rmd s\\
 & \quad-q(s)\frac{\partial p}{\partial s}(\theta(s))\rmd s+p(\theta(s))\rmd q(s)\\
 & =q(s)\left(\frac{\partial(pC^{i})}{\partial u^{i}}(\eta(s))-p(\theta(s))\frac{\partial C^{i}}{\partial u^{i}}(\eta(s))-\frac{\partial p}{\partial x^{i}}(\theta(s))Y^{i}(s)+\nu(\Delta_{x}p)(\theta(s))-\frac{\partial p}{\partial s}(\theta(s))\right)\rmd s\\
 & \quad+\sqrt{2\nu}q(s)\frac{\partial p}{\partial x^{i}}(\theta(s))\rmd M^{i}(s)\\
 & =\sqrt{2\nu}q(s)\frac{\partial p}{\partial x^{i}}(\theta(s))\rmd M^{i}(s),
\end{align*}
it follows that $(Z(s))_{s\in [0,t]}$ is a local martingale 
\[
Z(s)=Z(0)+\sqrt{2\nu}\int_{0}^{s}q(r)\frac{\partial p}{\partial x^{i}}(\theta(r)) \rmd M^{i}(r),
\]
which associates with an increasing sequence of stopping times $\{\tau_n\}_{n\geq 0}$
\begin{align*}
    \tau_n=\min\left\{t,\,\inf\{s\geq 0\,:\,|(Y(s),X(s))-(u,x)|\geq n\}\right\}.
\end{align*}
After taking expectation and utilising the continuity of $p$, we apply the dominated convergence theorem to obtain 
\begin{align*}
    p(u;x,t)
    =Z(0)
    =\lim_{n\rightarrow \infty} \E\left[Z(\tau_n)\right] =\E\left[Z(t)\right],
\end{align*}
which coincides with the representation formula. 

Regarding the second part \ref{claim: viscosity_solution}, we introduce functions $C^i_\varepsilon(x,u,t)$ and $p_{0,\varepsilon}(u;x)$, which are smooth functions (e.g. obtained from convolution with mollifiers) and converge to $C,p_0$ uniformly on compact sets. We further define the system of non-degenerate SDEs on the time interval $[0,t]$
\begin{align*}
    \rd X^i_\varepsilon(s) &= -Y^i_\varepsilon(s)\rmd t+\sqrt{2\nu}\rmd M^i(s), \qquad\qquad \qquad X_\varepsilon(0)=x,\\
    \rd Y^i_\varepsilon(s) &=C^i_\varepsilon(X^i_\varepsilon(s), Y^i_\varepsilon(s),t-s)\rmd t+\sqrt{\varepsilon}\rmd M^i(s), \;\, Y_\varepsilon(0)=u,
\end{align*}
and bounded process $(q_{\varepsilon}(s))_{s\in [0,t]}$
\begin{align*}
    \rd q_{\varepsilon}(s)=q_{\varepsilon}(s)\frac{\partial C_{\varepsilon}^{k}}{\partial u^{k}}(X(s),Y(s),t-s)\rmd s,\quad q_{\varepsilon}(0)=1.
\end{align*}
Consider the following parabolic PDE
\begin{align*}
    \left(\frac{\partial}{\partial t}+u^{i}\frac{\partial}{\partial x^{i}}-\nu\Delta_{x}-\varepsilon\Delta_u \right)p_\varepsilon & =\frac{\partial}{\partial u^{i}}\left(p_{\varepsilon} C_\varepsilon^{i}\right),\\
    p_\varepsilon(u;x,t)&=p_{0,\varepsilon}(u;x),
\end{align*}
it admits a unique classical smooth solution $p_{\varepsilon}$ by classical PDE theory. Moreover, the solution possesses the representation 
\begin{align*}
    p_{\varepsilon}(u;x,t)=\E\left[p_{0,\varepsilon}(Y_{\varepsilon}(t),X_{\varepsilon}(t))q_{\varepsilon}(t)\right],
\end{align*}
if we make use of 1). Applying Burkholder-Davis-Gundy and Gronwall inequalities (or following routine arguments in \citep{Kloeden1992NSSDE}), we have
\begin{align*}
    \E\left[\sup_{0\leq s\leq t}\left|(Y_\varepsilon(s),X_\varepsilon(s))-(Y(s),X(s))\right|^2\right]\rightarrow 0 
\end{align*}
as $\varepsilon\rightarrow 0$. Therefore, at least through a subsequence, $(Y_\varepsilon(s),X_\varepsilon(s))\rightarrow (Y(s),X(s))$ almost surely, leading to 
\begin{align*}
    p_{\varepsilon}(u;x,t)=\E\left[p_{0,\varepsilon}(Y_{\varepsilon}(t),X_{\varepsilon}(t))q_{\varepsilon}(t)\right]\rightarrow \E\left[p_{0}(Y(t),X(t))q(t)\right]=p(u;x,t), 
\end{align*}
uniformly on compact sets. Therefore equation (\ref{eq: S_rep_h}) is a viscosity solution by Proposition 5.8 in \citep{yong1999stochastic}, whereas the uniqueness follows from \citep{ishii1990viscosity}.
\end{proof}
The PDF PDE (\ref{eq: PDF_PDE_homogeneous}) boils down to a degenerate parabolic PDE after the weak homogeneity has been applied. Apart from solving the PDE, the stochastic representation (\ref{eq: S_rep_h}) offers a route to numerically solving the PDE. The PDE (\ref{eq: PDF_PDE_homogeneous}) has 6 dimensions in space and 1 dimension in time, which is a challenging task for classical finite difference methods due to the size of grid in space. Instead, we can simulate the solution based on Monte-Carlo methods directly.
\begin{rem}
Moreover, the stochastic representation formula can be extended to non-weakly homogeneous turbulent flows. Indeed,
\begin{align*}
p(u;x,t) & =\E\left[p_{0}(Y(t);X(t))q(t)\right],
\end{align*}
is a solution to (\ref{eq: PDF_PDE_origin}) subject to the system
of SDEs 
\begin{align*}
\rd Y^{i}(s) & =C^{i}(X(s),Y(s),t-s)\rmd s, \quad Y(0)=u,\\
\rd X^{i}(s) & =-Y^{i}(s)\rmd s+\sqrt{2\nu}\rmd M^{i}(s),\quad X(0)=x,
\end{align*}
with $C^{i}=\nu\frac{\partial B_{k}^{i}}{\partial x^{k}}-\nu A^{i}+Q^{i}$ for
all $i\in\{1,2,3\}$ if we impose the following condition on $p$:
\begin{align*}
\partial_{u^{i}}\left(B^{i}(x,u,t)\cdot\nabla_{x}p(u;x,t)\right) & =0.
\end{align*}

We are now in a position to verify the solution of our PDF PDE (\ref{eq: PDF_PDE_homogeneous}) is indeed a PDF. That is, it must carry two properties, including positivity and the mass preservation property. It turns out under some technical assumptions, the mass preservation property is equivalent to the divergence-free condition.
\end{rem}

\begin{lem}\label{lem: weak_homogeneous_properties}
Let $C$ be a given function which satisfies assumption \normalfont{\textbf{\ref{assumption: LipschitzCondition}}}, while $p_0(u;x)$ satisfies assumption \normalfont{\textbf{\ref{assumption: Polygrowth}}}. Let $p(u;x,t)$ be to the solution to equation (\ref{eq: PDF_PDE_homogeneous}), we have the following statements:
\begin{enumerate}[label=\arabic*)]
    \item If $p_{0}\geq0$ then $p(u;x,t)\geq0$.
    \item \label{claim: vanishing_boundary} If we further assume there exists
    $m\geq1$ such that $\lim_{|u|\rightarrow\infty}p_{0}(u;x)|u|^{m}=0$ uniformly in $x$ as $|u|\rightarrow\infty$, then $\lim_{|u|\rightarrow\infty}p(u;x,t)|u|^{m}=0$
    uniformly in $x$. Moreover, if $m>q+3$ for some $q\geq1$, then
    $\intr|u|^{q}p(u;x,t)\rmd u<\infty$.
    \item Suppose $p_{0}$ also satisfy the conditions in  \ref{claim: vanishing_boundary} and $\intr p_{0}(u;x)\rmd u=1$, we have $\int_{\R^{3}}p(u;x,t)\rmd u=1$ for all $x,t$ if and only if 
    \begin{equation}
        \frac{\partial}{\partial x^{i}}\int_{\R^{3}}u^{i}p(u;x,t)\rmd u=0.\label{eq: div-free for weak flow}
    \end{equation}
\end{enumerate}
\end{lem}

\begin{proof}
1) follows directly from the stochastic representation formula (\ref{eq: S_rep_h}). To deal with \ref{claim: vanishing_boundary}, let $Y(s;u)$ and $X(s;u)$ denote the solutions to the SDEs in Theorem \ref{thm: weakly_homogeneous}, while we put $u$ in the representation on $X,Y$ to emphasize their dependence on the initial data $u$. As a result of the Lipschitz condition, we deduce that $|C(X(s;u),Y(s;u),t-s)\cdot Y(s;u)|\leq\frac{1}{2}K(1+|Y(s;u)|^{2})$ for some $K>0.$ Introducing scalar processes $Z(s;u)\coloneqq|Y(s;u)|^{2}$, $\alpha(s;u)$ and $\beta(s;u)$:
\begin{align*}
\rd Z(s;u) & =2C(Y(s;u),X(s;u),t-s)\cdot Y(s;u)\rmd s,\\
\rd\alpha(s;u) & =-K(1+\alpha(s;u))\rmd s,\\
\rd\beta(s;u) & =K(1+\beta(s;u))\rmd s,
\end{align*}
with $\alpha(0;u)=\beta(0;u)=|Y(0;u)|^{2}=|u|^{2}$, $Z(s;u)$ must satisfy the inequality $\alpha(s;u)\leq Z(s;u)\leq\beta(s;u)$, since we are able to make use of the comparison theorem (see for example, \citep{McNabb1986}) and 
\begin{align*}
-K(1+Z(s;u)) & \leq2C(Y(s;u),X(s;u),t-s)\cdot Y(s;u)\leq K(1+Z(s;u))
\end{align*}
for all $Z(s;u)\geq0$. Therefore
\[
Z(s;u)\in[\exp(-Ks)(|u|^{2}+1)-1,\exp(Ks)(|u|^{2}+1)-1].
\]
Applying the dominated convergence theorem, we get
\begin{align*}
\lim_{|u|\rightarrow\infty}|p(u;x,t)||u|^{m} & \leq\exp(Kt)\E\left[\lim_{|u|\rightarrow\infty}|p_{0}(Y(t;u);X(t,u))||u|^{m}\right]=0.
\end{align*}
Recall that $p_{0}(u;x)$ has at most polynomial growth in $(u,x)$, i.e. $p_{0}(u;x)\leq L(|u|^{l}+|x|^{l})$ for some $l,L\geq1$. Denoting $B(0,R)$ as the ball centred at the origin with radius $R>0$, the moment bound can be obtained by 
\begin{align*}
\intr|u|^{q}p(u;x,t)\rmd u & \leq\int_{B(0,R)}|u|^{q}\exp\left(Kt\right)\E\left[|Y(s;u)|^{l}+|X(s;u)|^{l}\right]\rmd u\\
 & \quad+\int_{\R^{3}\backslash B(0,R)}|u|^{m}p(u;x,t)\frac{1}{|u|^{m-q}}\rmd u\\
 & <\infty,
\end{align*}
where $R$ is chosen large enough such that $|u|^{m}p(u;x,t)\leq K_{1}$ for all $|u|>R$ and some constant $K_{1}>0$. 

In order to check 3), we consider $f(x,t)\coloneqq\int_{\R^{3}}p(u;x,t)\rmd u$. Integrating the PDF PDE with respect to the variable $u$, we obtain that 
\begin{align*}
\left(\frac{\partial}{\partial t}-\nu\Delta_{x}\right)f(x,t) & =-\frac{\partial}{\partial x^{i}}\int_{\R^{3}}u^{i}p(u;x,t)\rmd u+\int_{\R^{3}}\nabla_{u}\cdot\left(pC\right)\rmd u,\\
f(x,0) & =1.
\end{align*}
Therefore, $|pC|\rightarrow0$ and the conclusion follows immediately.
\end{proof}
\begin{rem}
There exists a (viscosity) solution to equation (\ref{eq: PDF_PDE_homogeneous}) when we assume $C$ is Lipschitz in $(x,u)$ and uniformly continous in $t$, but the fact that $C^{i}$ is bounded in $x$ for all $i\in\{1,2,3\}$ is crucial to part \ref{claim: vanishing_boundary} of Lemma \ref{lem: weak_homogeneous_properties}. Assuming $C^{i}(y,x,t)=2a^{i}y^{i}+b^{i}x^{i}+c^{i}(t)$ for some $a^{i},b^{i}\in\R,$ $(a^{i})^{2}>b^{i}$ and bounded function $c^{i}$, where no Einstein's convention is applied, the solution to the system of SDEs $(Y^{i},X^{i})^{T}$ is given by
\begin{align*}
\left[\begin{array}{c} Y^{i}(s)\\
                        X^{i}(s)\end{array}\right] 
& =\exp\left(L^{i}s\right)\left[\begin{array}{c}
u^{i}\\
x^{i}
\end{array}\right]+\int_{0}^{s}\exp\left(L^{i}(s-r)\right)\left[\begin{array}{c}
c^{i}(r)\\
0
\end{array}\right]\rmd r+\int_{0}^{s}\exp\left(M^{i}(s-r)\right)\left[\begin{array}{c}
0\\
\sqrt{2\nu}
\end{array}\right]\rmd M_{r}^{i},
\end{align*}
with $L^{i}\coloneqq\left[\begin{array}{cc} 2a^{i} & b^{i}\\-1 & 0 \end{array}\right]$. The corresponding eigenvalues of $L^{i}$ are $\lambda_{1,i}=a^{i}+\sqrt{(a^{i})^{2}-b^{i}}$ and $\lambda_{2,i}=a^{i}-\sqrt{(a^{i})^{2}-b^{i}}$, which implies
\begin{align*}
\exp\left(L^{i}t\right) & =\frac{1}{\lambda_{2,i}-\lambda_{1,i}}\left[\begin{array}{cc}
-\lambda_{1,i}\exp(\lambda_{1,i}t)+\lambda_{2,i}\exp(\lambda_{2,i}t) & b^{i}\left(-\exp(\lambda_{1,i}t)+\exp(\lambda_{2,i}t)\right)\\
\exp(\lambda_{1,i}t)-\exp(\lambda_{2,i}t) & \lambda_{2,i}\exp(\lambda_{1,i}t)-\lambda_{1,i}\exp(\lambda_{2,i}t)
\end{array}\right]
\end{align*}
and $Y^{i}(s)$ is independent of $u^{i}$ provided $s=\frac{\ln(\lambda_{2,i}/\lambda_{1,i})}{\lambda_{1,i}-\lambda_{2,i}}>0$. In particular, $a^{i}<0$ and $b^{i}>0$ with $(a^{i})^{2}>b^{i}$ lead to $s>0$, and therefore 
\[
p(u;x,t)=\E\left[p_{0}(Y(t);X(t))q(t)\right]
\]
not necessarily vanishes at $|u|\rightarrow\infty$ for all $(x,t)$ if we only assume $p_{0}(u;x)\to0$ as $|u|\rightarrow\infty$ i.e. $\intr p(u;x,t)\rmd u =\infty$ for some $t>0$.
\end{rem}

We introduce a sufficient condition for the mass-preservation property when $C^{i}$ satisfies the constraints in the following Lemma.
\begin{lem}
Suppose $C$, $p_{0}$ satisfy condition \ref{claim: vanishing_boundary} in Lemma \ref{lem: weak_homogeneous_properties}
with $m\geq2$ and define $\{U(x,t)\}_{x,t}$ to be the corresponding random field such that $U(x,t)$ has the PDF $p(u;x,t)$ for all $(x,t)$. If equation (\ref{eq: div-free for weak flow}) is satisfied, we have for all $(x,t)\in\R^{3}\times\R_{+}$,
\begin{align*}
\nabla_{x}\cdot\E\left[C(x,U(x,t),t)\right] & =-\partial_{x^{i}}\partial_{x^{j}}\E\left[U^{i}(x,t)U^{j}(x,t)\right].
\end{align*}
\end{lem}

\begin{proof}
We apply $u\cdot\nabla_{x}$ on both sides of the PDE (\ref{eq: PDF_PDE_homogeneous})
followed by integrating with respect to $u$, resulting
\begin{align*}
\int_{\R^{3}}(u\cdot\nabla_{x})\left(\partial_{t}-\nu\Delta_{x}\right)p(u;x,t)\rmd u & =\left(\partial_{t}-\nu\Delta_{x}\right)\nabla_{x}\cdot\int_{\R^{3}}up(u;x,t)\rmd u=0
\end{align*}
as well as 
\begin{align*}
\int_{\R^{3}}(u\cdot\nabla_{x})u\cdot\nabla_{x}p(u;x,t)\rmd u & =\int_{\R^{3}}u^{i}u^{j}\partial_{x^{i}}\partial_{x^{j}}p(u;x,t)\rmd u\\
 & =\partial_{x^{i}}\partial_{x^{j}}\E\left[U^{i}(x,t)U^{j}(x,t)\right].
\end{align*}
Eventually, the right-hand side can be written in terms of 
\begin{align*}
\int_{\R^{3}}(u\cdot\nabla_{x})\nabla_{u}\cdot\left(p(u;x,t)C(x,u,t)\right)\rmd u & =\int_{\R^{3}}\partial_{x^{i}}\nabla_{u}\cdot\left(u^{i}p(u;x,t)C(x,u,t)\right)\rmd u\\
 & \quad-\int_{\R^{3}}\nabla_{u}\cdot\left(p(u;x,t)C(x,u,t)\right)\rmd u\\
 & =-\nabla_{x}\cdot\E\left[C(x,U(x,t),t)\right],
\end{align*}
provided $|p(u;x,t)C(x,u,t)u^{i}|\rightarrow0$ as $|u|\rightarrow\infty$. This has been guaranteed by the growth condition on $C^{i}$.
\end{proof}
If we further assume the incompressible viscous turbulent fluid flow is both weakly homogeneous and weakly isotropic, the PDF PDE is a parabolic-transport equation 
\begin{align}
\left(\frac{\partial}{\partial t}+u^{i}\frac{\partial}{\partial x^{i}}-\nu\Delta_{x}\right)p & =\frac{\partial}{\partial u^{i}}\left(pC^{i}\right),\label{eq: PDF_PDE_isotropic}\\
p(u;x,0) & =p_{0}(u;x),\nonumber 
\end{align}
where $C$ is essentially 
\begin{align*}
C^{i}(x,u,t) & =-\nu A^{i}(u,t)+Q^{i}(x,u,t)\\
 & =-\nu A^{i}(u,t)+\int_{\R^{3}}\frac{y^{i}}{4\pi|y|}\frac{\partial^{2}\left(\sigma^{jk}+b^{j}b^{k}\right)}{\partial y^{k}\partial y^{j}}(y,u,t)\rmd y
\end{align*}
depending only on $u$ and $t$ but not on $x$. 
\begin{cor}
Let $C^{i}(u,t)$ be Lipschitz continuous in $u$ and uniformly continuous in $t$, for all $i\in\{1,2,3\}$. For $t>0$ and $(x,u)\in\R^{3}\times\R^{3}$, $(X,Y,q)$ denotes the unique solution to the system of equations for all $i\in \{1,2,3\}$ and $s\in[0,t]$: 
\begin{align*}
\rd X^{i}(s) & =-Y^{i}(s)\rmd s+\sqrt{2\nu}\rmd M_{s}^{i},\quad X(0)=x,\\
\rd Y^{i}(s)&=C^{i}(Y(s),t-s)\rmd s,\quad Y(0)=u,
\end{align*}
and 
\[
\rd q(s)=-q(s)\frac{\partial C^{k}}{\partial u^{k}}(Y(s),t-s)\rmd s,\quad q(0)=1.
\]
Suppose $p(u;x,t)$ is a $\mathcal{C}^{1,1,1}(\R^{3}\times\R^{3}\times[0,T])$ solution to (\ref{eq: PDF_PDE_isotropic}) then 
\[
p(u;x,t)=q(t)\int_{\R^{3}}H(x,t,z;u)p_{0}(Y(t);z)\rmd z,
\]
where 
\[
H(x,t,y;u)=\frac{1}{(4\pi\nu t)^{3/2}}\exp\left(-\frac{\left|y-x+\int_{0}^{t}Y(s)\rmd u\right|^{2}}{4\nu t} \right).
\]
Therefore if $p_{0}(u;x)\geq0$ for every $x$, then so is $p(u;x,t)$ for any $(x,t)$. If $p_{0}(u;x)$ is a PDF for all $x$, then $p(u;x,t)$ is again a PDF for all $(x,t)$ if only if the following constraint holds: 
\[
\frac{\partial}{\partial x^{i}}\int_{\R^{3}}u^{i}p(u;x,t)\rmd u=0.
\]
\end{cor}

\begin{proof}
By definition, $Y$ and $q$ are deterministic processes, while $X$ is a $3$-dimensional Gaussian process such that for all $E \in \mathscr{B}(\R^3)$ and $s\in [0,t]$,
\begin{align*}
    \P(X_s \in E) = \int_E H(x,s,z;u)\rmd z.
\end{align*}
As a result of Theorem \ref{thm: weakly_homogeneous}, 
\begin{align*}
p(u;x,t)  =q(t)\mathbb{E}\left[p(Y(t);X(t),0)\right]
  =q(t)\int_{\R^{3}}H(x,t,z;u)p_{0}(Y(t);z)\rmd z.
\end{align*}
The remaining part is a consequence of Lemma \ref{lem: weak_homogeneous_properties}.
\end{proof}

\subsection{Inviscid flows}

The modelling of inviscid incompressible flows is significantly simplified comparing to modelling viscid flow, since the PDF PDE can be solved without imposing the weak homogeneity or the weak isotropy conditions. As $\nu=0$, the velocity $U(x,t)$ fulfils the Euler equations, while the PDF PDE becomes the following transport equation 
\begin{align}
\begin{split}\label{eq: PDF_PDE_inviscid}
    \frac{\partial p}{\partial t}+u^{i}\frac{\partial p}{\partial x^{i}} & =\frac{\partial}{\partial u^{i}}\left(pQ^{i}\right),\\
    p(u;x,0) & =p_{0}(u;x).
\end{split}
\end{align}

\begin{thm}
\label{thm: PDF_PDE_inviscid}Suppose that $Q^{i}(x,u,t)$ satisfies assumption \normalfont{\textbf{\ref{assumption: LipschitzCondition}}}, $p_0(u;x)$ satisfies assumption \normalfont{\textbf{\ref{assumption: Polygrowth}}} and $p\in \mathcal{C}^{1,1,1}(\R^3, \R^3, [0,T])$ for some fixed $T>0$ is a solution to the transport PDE (\ref{eq: PDF_PDE_inviscid}), we have
\begin{equation}
p(u;x,t)=p_0(Y(t);X(t))q(t)\label{eq: rep_inviscid}
\end{equation}
for every $t>0$, $x$ and $u$, where $(X,Y,q)$ is the unique solution to the following system of ODEs: 
\begin{align}
\begin{split}\label{eq: ODEsystem}
    \rd X^{i}(s) & =-Y^{i}(s)\rmd s,\quad X(0)=x, \\
    \rd Y^{i}(s) & =Q^{i}(X(s),Y(s),t-s)\rmd s,\quad Y(0)=u\\
    \rd q(s)&=q(s)(\nabla_{u}\cdot Q)(X(s),Y(s),t-s)\rm \rmd s,\quad q(0)=1.
\end{split}
\end{align}
for $i=1,2,3$. Moreover, if $p(u;x,0)\geq0$ for every $x$, then so is $p(u;x,t)$ for any $(x,t)$. If $p(u;x,0)$ is a PDF for all $x$ and $p_{0}(u;x)|u|^{m}\rightarrow0$ uniformly in $x$ for some $m\geq1$, then $p(u;x,t)$ is again a PDF for all $(x,t)$ if only if the the following constraint is satisfied:
\[
\frac{\partial}{\partial x^{i}}\intr u^{i}p(u;x,t)\rmd u=0.
\]
\end{thm}

\begin{proof}
The system of ODEs (\ref{eq: ODEsystem}) has a unique solution pair $(X,Y)$ and 
\[
q(s)=\exp\left(\int_{0}^{s}\frac{\partial Q^{i}}{\partial u^{i}}(X(r),Y(r),t-r)\rmd r\right),
\]
which is a bounded process by the Lipschitz assumption. Let $\theta(s)=(Y(s);X(s),t-s)$, $\eta(s)=(X(s),Y(s),t-s)$ and define
\[
h(s)=p(\theta(s))q(s),
\]
then $h(0)=p(u;x,t)$ and $h(t)=p(Y(t);X(t),0)q(t)$. Moreover for
$s\in[0,t]$, we have 
\begin{align*}
h'(s) & =q(s)\left(\frac{\partial p}{\partial u^{i}}(\theta(s))\frac{\partial Y^{i}}{\partial s}+\frac{\partial p}{\partial x^{i}}(\theta(s))\frac{\partial X^{i}}{\partial s}-\frac{\partial p}{\partial s}(\theta(s))\right)+p(\theta(s))\frac{\partial q}{\partial s}\\
 & =q(s)\left(\frac{\partial p}{\partial u^{i}}(\theta(s))Q^{i}(\eta(s))+\frac{\partial p}{\partial x^{i}}(\theta(s))Y^{i}(s)-\frac{\partial p}{\partial s}(\theta(s))+p(\theta(s))\frac{\partial Q^{i}}{\partial u^{i}}(\eta(s))\right)\\
 & =0,
\end{align*}
hence $h$ is constant on $[0,t]$ and we make use of $p(u;x,t)=h(t)$ to deduce (\ref{eq: rep_inviscid}). Regarding the positivity and mass preservation properties, the proof is almost the same as the proof of Lemma \ref{lem: weak_homogeneous_properties}, except (\ref{eq: ODEsystem}) is not stochastic: if $p_{0}\geq0$ then $p(u;x,t)\geq0$ directly and 
\begin{align*}
\lim_{|u|\rightarrow\infty}|p(u;x,t)||u|^{m} & \leq\sup_{s\in[0,t]}|q(s)|\lim_{|u|\rightarrow\infty}|p_{0}(Y(t;u);X(t;u))||u|^{m}=0.
\end{align*}
Let $f(x,t)=\int_{\R^{3}}p(u;x,t)\rmd u$. If $p_{0}$ is a probability density, then $f(x,0)=1$ for all $x$. By integrating the equation (\ref{eq: PDF_PDE_inviscid}) with respect to $u$, we obtain 
\[
\frac{\partial}{\partial t}f(x,t)-\frac{\partial}{\partial x^{i}}\int_{\R^{3}}u^{i}p(u;x,t)\rmd u=0,
\]
which leads us to the conclusion. 
\end{proof}

\section{Modelling the PDF: a concrete example}

On one hand, from the modelling point of view, only those solutions to the PDF PDE which satisfy the natural mass conservation condition (\ref{eq: div-free for weak flow}) can be used as models of distributions of turbulent velocity fields. On the other hand, from view-point of PDEs,  the mass conservation  property of solutions to the PDF PDE imposes a strong constraint on its solutions. As a matter of fact, most solutions of PDF PDE with given $A, B, Q$ do not satisfy this constraint. In general, solutions to the PDF PDE do not have an explicit expression, although our stochastic representations established in the previous sections may be helpful in dealing with the mass conservation. Consequently, it brings a next level of difficulty to verify the constraint (\ref{eq: div-free for weak flow}). It turns out that the natural condition that $\int_{\mathbb{R}^{3}}p(u;x,t)\rmd u=1$ for a solution $p(u;x,t)$ to the PDF PDE, which is equivalent to the mass conservation, is a very strict constraint for whatever the coefficients $A,B$ and $Q$ which may be modelled or measured. At least our experience demonstrates that the PDF solutions to the PDF PDE are rare, and indicates that the PDF solution $p(u;x,t)$ is not so sensitive for the choices of $A,B$ and $Q$, although we are unable to prove this claim in the present paper. Therefore the PDF PDE together with the mass conservation constraint is very rigid, and hence is good for modelling the PDF of turbulent flows. The authors hope to see further exploration in this direction in the future. 

In this section, we study an explicit example to the PDF PDE, where the mixed-derivative term vanishes in the PDE i.e. the turbulent flow is weakly homogeneous. The example seems artificial, but as we have explained above, we believe that this example has relevance to real turbulent flows.

\subsection{Space homogeneous density with perturbation}

For simplicity, the viscosity parameter in this example is set to be $\nu=1$. As the solution to the PDE (\ref{eq: PDF_PDE_homogeneous}) is solely determined by the initial data $p_{0}(u;x)$ as well as the function $C(x,u,t)$, we consider the simplest scenario $C=0$. Meanwhile, instead of setting up a common distribution such as Gaussian or exponential distribution for the initial data, we introduce the following non-negative function: for every $x\in \mathbb{R}^3$, let
\begin{align}
p_{0}(u;x)  =\alpha(u)+\beta(u)\gamma(x),
\label{eq: PDFexample}
\end{align}
where $\alpha(u)$ is a PDF given by
\begin{alignat*}{1}
\alpha(u) & =\frac{1}{(\sqrt{2\pi})^{3}\sqrt{\det\sigma}}\exp\left(-\frac{1}{2}u^{T}\sigma^{-1}u\right)
\end{alignat*}
with $\sigma_{ij}=(\frac{3}{2})^{i-2}\ione_{\left\{ i=j\right\} }$, corresponding to the PDF of a centred Gaussian vector with independent components. Here, $\beta$ satisfies $\int_{\R^{3}}\beta(u)\rmd u=0$, and is chosen to be the product of reciprocals
\begin{align*}
\beta(u) & =\begin{cases}
\prod_{i}\frac{1}{u^{i}}, & (u^{1},u^{2},u^{3})\in I,\\
0, & \text{otherwise},
\end{cases}
\end{align*}
and truncated if $u$ left the region $I$, where $I$ is defined as
\begin{align*}
I\coloneqq & \left(\left[\frac{1}{4},1\right]\cup\left[-1,-\frac{1}{4}\right]\right)\times\left(\left[\frac{1}{4},2\right]\cup\left[-2,-\frac{1}{4}\right]\right)\times\left(\left[\frac{2}{7},3\right]\cup\left[-3,-\frac{2}{7}\right]\right).
\end{align*}
 We further set the last function $\gamma$ as
\begin{align*}
\gamma(x) & =\frac{1}{36}\frac{1}{(\sqrt{2\pi})^{3}}\left\{ \left[\frac{30(x^{1}-1)}{30+3(x^{2})^{2}+2(x^{3})^{2}}\exp\left(-\frac{(x^{1})^{2}}{3}+\frac{1}{3}\sin\left(\sum_{i}x^{i}\right)\right)+\frac{\cos(x^{2}+x^{3})}{(|x|^{2}+1)}\right]+2\exp\left(-\frac{|x|^{2}}{200}\right)\right\} ,
\end{align*}
which vanishes at $|x|\rightarrow\infty$. We select these functions so that the positivity and mass preserving properties are fulfilled, and the initial velocity is a random field whose marginal density at $x$ is given by \eqref{eq: PDFexample}.


By the stochastic representation formula, the PDF of the random fields at $(x,t)$ has the form
\begin{align}
p(u;x,t) & =\alpha(u)+\beta(u)\E\left[\gamma\left(x-ut+\sqrt{2\nu}M_{t}\right)\right].\label{eq:ExplicitPDF}
\end{align}

\begin{figure}
\includegraphics[scale=0.55]{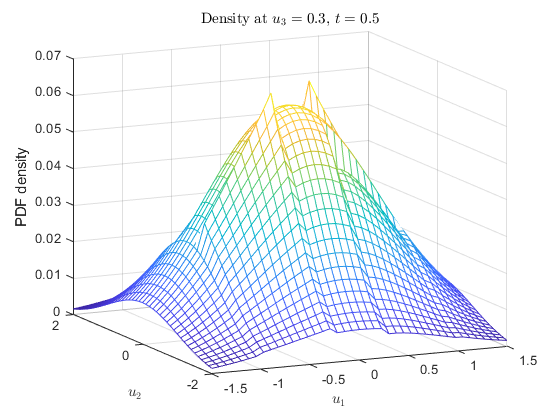}\includegraphics[scale=0.55]{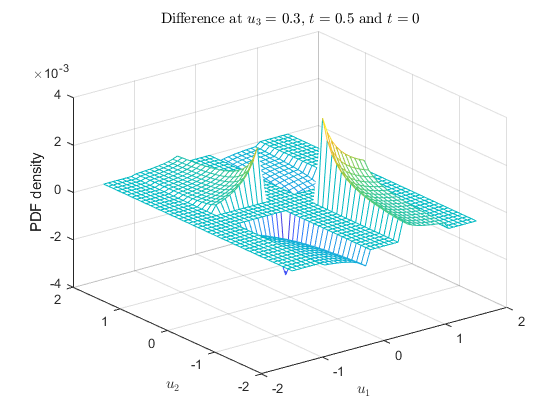}

\caption{$x=(0,0,0),\;t=\frac{1}{2}$\label{fig:T05}}
\end{figure}

We focus on the PDF in the random field $\{U(x,t)\}_{x,t}$ at $x=(0,0,0)$ and plot the graph of $u$ against $p(u;x,t)$ at $u_{3}=0.3$ for different time $t$. At $t=\frac{1}{2}$, $p(u;x,t)$ is discontinuous on the boundary of $I$ as in figure \ref{fig:T05}. If we compare the density of $p(u;x,\frac{1}{2})$ and $p(u;x,0)$ by evaluating $p(u;x,\frac{1}{2})-p(u;x,0)$, we can see the change of density in the region $I$. Meanwhile, the discontinuity becomes less apparent on the plot when we increase the time to $t=40$. From figure \ref{fig:T40}, the PDF of $U(x,t)$ is close to the density of a Gaussian random variable with density function $\alpha$, even if the discontinuity still exists. This is due to $\gamma$ vanishes at infinity and $\E[\gamma(X_{t})]\rightarrow0$ as $t\rightarrow\infty$. However, the impact of $\beta$ does not disappear from the velocity field. There is a strong discontinuity near $(\frac{1}{4},\frac{1}{4},\frac{2}{7})$ on $p(u;x,t)$ when $x=(12,12,12)$ and $t=40$. The PDF at $x=(12,12,12)$ is asymmetric and has a different evolution than the density at the origin, which demonstrates that the impact $\beta$ shifts from the origin to somewhere far away as time changes.

\subsection{Motivation for the construction}

The mass-preserving property of the PDE (\ref{eq: PDF_PDE_homogeneous}) corresponds to the divergence-free condition (\ref{eq: div-free for weak flow}), which is difficult to verify explicitly even when we have the stochastic representation. Apart from describing the motivation for choosing $\alpha,\beta$ and $\gamma$, we will demonstrate the our solution to the PDF PDE (\ref{eq:ExplicitPDF}) satisfies the divergence-free constraint.

Assuming $C=0$ does not imply that the turbulence flow associated to solution $p(u;x,t)$ is weakly isotropic. For example, we force the conditional average increment to satisfy \begin{align*}
\rho^{i}(x,y,u,t) & =O(|y-x|^{3})
\end{align*}
when $|y-x|\to0$ and to vanish sufficiently fast when $|y|\rightarrow\infty$, which naturally leads to $A_{jk}^{i}=0$. In addition, we let 
\begin{align*}
\rho^{i}(x,y+x,u,t) & =\rho^{i}(x,-y+x,u,t)
\end{align*}
for all $i$. If the conditional variance is of the form $\sigma^{jk}(x,y,u,t)=c+\prod_{i}f_{i,j,k}(y_{i}-x_{i},t)\lambda(x,u,t)$, we conclude that 
\begin{align*}
Q^{i}(x,u,t) & =\sum_{j,k}\int\frac{1}{4\pi|y-x|^{3}}\frac{\partial}{\partial y^{i}}\frac{\partial^{2}\left(\sigma^{kl}+b^{k}b^{l}\right)}{\partial y^{k}\partial y^{l}}\rmd y\\
 & =\sum_{j,k}\int\frac{1}{4\pi|y|^{3}}\frac{\partial^{3}}{\partial y^{i}\partial y^{j}\partial y^{k}}\left(\prod_{i}f_{i,j,k}(y_{i})\lambda(x,u,t)+\left(\rho^{j}(x,y+x,u,t)\rho^{k}(x,y+x,u,t)\right)\right)\rmd y\\
 & =0,
\end{align*}
provided that $f_{i,j,k}$ is an even function with $f_{i,j.k}(z)\rightarrow0$ sufficiently fast as $z\rightarrow0$ and $z\rightarrow\infty$ for all $t$. In particular $f_{i,j,k}(z,t)=\frac{1}{8}z^{4}\exp\left(-\frac{z^{2}(1+t)}{2}\right)$ fits the criterion that we required. A pair of the conditional statistics $\rho$, $\sigma$ obeying the above constraints is a reasonable choice leading to $C=0$.
\begin{figure}
\includegraphics[scale=0.55]{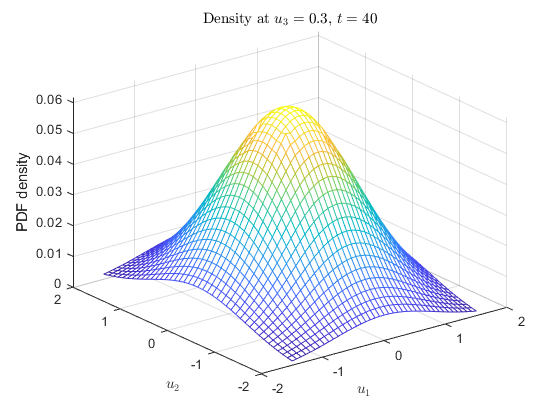}\includegraphics[scale=0.55]{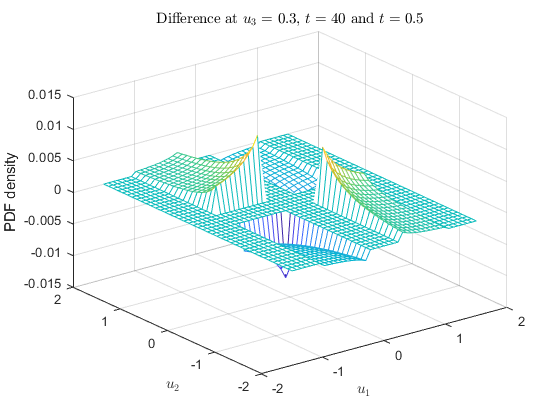}\caption{$x=(0,0,0),\;t=40$\label{fig:T40}}
\end{figure}

Regarding the initial data $p_{0}(u;x)$ of the form (\ref{eq: PDFexample}), there is no strong restriction on $\alpha$, hence $\alpha$ is allowed to be replaced by another PDF, which not necessarily corresponds to a Gaussian vector. The crucial ingredients in $p_{0}(u;x)$ are $\beta$ and $\gamma$, which ensure the solution $p(u;x,t)$ satisfy the mass conservation property. As a consequence of the fact that $C=0$, the right-hand side of equation (\ref{eq: PDF_PDE_isotropic}) vanishes, ending up an equation which depends solely on the derivatives of $p$ with respect to $(t,x)$. Moreover, the relevant SDEs in the stochastic representation (\ref{eq: S_rep_h}) have the explicit form $Y_{t}=u$ and $X_{t}=x-ut+\sqrt{2\nu}M_{t}$ , while the divergence-free constraint reads that
\begin{align*}
\nabla_{x}\cdot\int_{\R^{3}}u & p(u;x,t)\rmd u=\E\left[\int_{\R^{3}}u\cdot\nabla_{x}p(u,x+\sqrt{2\nu}M_{t}-ut,0)\rmd u\right]=0.
\end{align*}
for all $x\in\R^{3},t>0$. In particular, if
\begin{align}
\int_{\R^{3}}u\cdot\nabla_{x}p(u,x-ut,0)\rmd u & =0\label{eq: suffice_condition_div_free_C0}
\end{align}
is ensured for all $(x,t)$, $\int_{\R^{3}}p(u;x,t)\rmd u=1$ is guaranteed.

When $t>0,$ the left-hand side of the equation (\ref{eq: suffice_condition_div_free_C0}) has the following form
\begin{align*}
\int_{\R^{3}}u\cdot\nabla_{x}p(u;x-ut,0)\rmd u & =\int_{\R^{3}}u\cdot\nabla_{x}\gamma(x-ut)\beta(u)\rmd u\\
 & =-\frac{1}{t}\int_{\R^{3}}u\cdot\nabla_{u}(\gamma(x-ut))\beta(u)\rmd u\\
 & =-\frac{1}{t}\int_{\R^{3}}\nabla_{u}\cdot\left(u\gamma(x-ut)\beta(u)\right)-\nabla_{u}\cdot\left(u\beta(u)\right)\gamma(x-ut)\rmd u.
\end{align*}
If $\gamma$ and $\beta$ decay sufficiently fast as $|u|\rightarrow\infty$ for all $x$, the first term in the integrand has zero contribution after integrated with respect to $u$. Apart from our consideration on the mathematical side, our choice of $\gamma$ guarantees the impact of $\beta$ vanish as $|x|\rightarrow\infty$, but the speed of decay is slow enough such that the impact of $\beta$ is still observable when $t$ is rather large. Last but not least, it must satisfy the following constraint
\begin{align*}
||\gamma(x)||_{\infty} & \leq\sup_{u\in\R^{3}}\bigg|\frac{\beta(u)}{\alpha(u)}\bigg|.
\end{align*}
We remark that under current form of $\beta$, $\gamma$ decays as $|x|\rightarrow0$ is not a necessity. 
\begin{figure}
\includegraphics[scale=0.55]{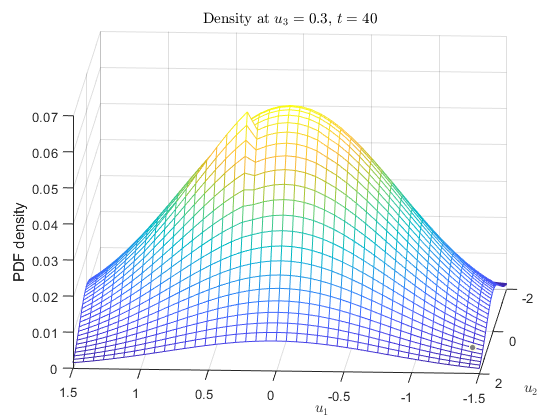}\includegraphics[scale=0.55]{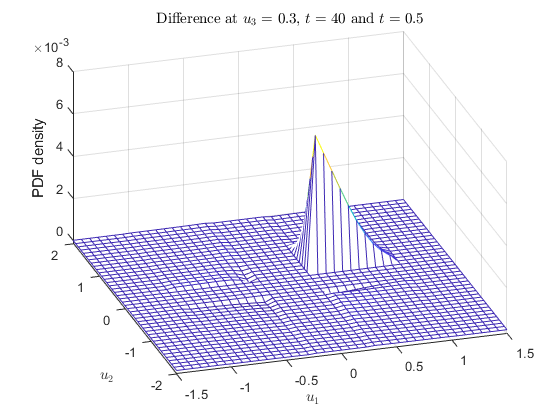}

\caption{$x=(12,12,12),t\;=40$\label{fig: T40X12}}
\end{figure}

The remaining problem turns into finding our right $\beta$ such that $\int_{\R^{3}}\nabla_{u}\cdot\left(u\beta(u)\right)\gamma(x-ut)\rmd u=0$ for all $t>0$ and $\intr u\beta(u)\rmd u\cdot\nabla_{x}\gamma(x)=0$, which ensure the initial data must also satisfy the divergence-free condition (\ref{eq: div-free for weak flow}). Our choice on $\beta$ is motivated by the fact that equation (\ref{eq: suffice_condition_div_free_C0}) is satisfied provided we impose the constraint $\nabla_{u}\cdot(u\beta(u))=0$ on $\beta$. $\beta$ is truncated to this form, in order to guarantee the integrability as well as the positivity of $p_{0}(u;x)$. Moreover, we can ensure $\int_{\R^{3}}\beta(u)\rmd u=0$, $\int_{\R^{3}}u^{i}\beta(u)\rmd u=0$ and $\partial_{u^{i}}(u^{i}\beta(u))=\partial_{u^{i}}(u^{j}u^{k})^{-1}=0$ for all $i=1,2,3$ and $i,j,k$ being distinct. Nevertheless, the purpose of symmetry of the interval $I$ and $\beta$ is to simplify our example, therefore $\beta$ can be asymmetrical.

The PDF (\ref{eq:ExplicitPDF}) at $(x,t)$ is discontinuous
in the variable $u$, but it is still a strong solution to the PDE, because the $\partial_{u^{i}}p$ disappears in the PDE (\ref{eq: PDF_PDE_isotropic}) in this circumstance.

\section{Concluding remarks}
This paper derives a new PDE which describes the evolution of one-time one-point PDF of the velocity random field of a turbulent flow. The PDF PDE, which is highly non-linear and is determined by two conditional statistics of a turbulent flow, should be a useful tool in modelling distributions of turbulence velocity fields.

The modelling of viscous turbulence in various environments by solving  numerically the PDF PDE (\ref{eq: PDF_PDE_origin_C}) with measured data or based on the priori determination of $A,B$ and $Q$ should be beneficial in understanding turbulent flows. To implement good models of PDFs for turbulent flows, we need to numerically calculate solutions of the PDF PDE, with fed data which determine the functions $A$, $B$ and $Q$. The solution has to satisfy the natural constraint, that the mass must be preserved through out the evolution of the PDF. The conservation of the total mass of the solution is an important topic itself and is worth of further study. Finally, we would like to point out that the coefficients $A$, $B$ and $Q$ defined in equation (\ref{eq: BA_in_PDE}), which determine the statistics of the turbulence at one-time one-space, must have significant physical meaning in turbulence. These coefficients, which are considered as turbulent flow parameters, should play their roles in further research. 


\end{document}